\documentclass[12pt]{amsart}
\usepackage{amsmath,amssymb,amsfonts,amsbsy,mathtools,cite,setspace}
\usepackage{graphicx}
\usepackage{amssymb}
\usepackage{amsthm}
\usepackage{enumerate}
\usepackage{wasysym}
\usepackage{cite}
\usepackage{tikz}

\newtheorem{defn}{Definition}
\newtheorem{prop}{Proposition}
\newtheorem{thm}{Theorem}

\newtheorem{lem}{Lemma}

\newtheorem{cor}{Corollary}

\newcommand{\C}{{\mathbb C}}
\begin{document}
\title[An Uncertainty Principle for Completely Positive Maps]{An Uncertainty Principle for Completely Positive Maps}
\author{Jeremy Levick}
\maketitle

\section{Uncertainty For CP Maps}
\subsection{Completely Positive Maps}
\begin{defn} A completely positive map $\Phi: M_n(\mathbb{C})\rightarrow M_m(\mathbb{C})$ is a linear map that preserves positivity, $a\geq 0 \Rightarrow \Phi(a)\geq 0$, and whose actions on the blocks of a $kn \times kn$ block matrix are positive for all $k \in \mathbb{N}$:\\
$$A = \sum_{i,j=1}^k E_{ij}\otimes A_{ij} \geq 0 \Rightarrow \sum_{i,j=1}^k E_{ij}\otimes \Phi(A_{ij})\geq 0 \ \ \forall k\in \mathbb{N}$$. 
\end{defn}

A theorem of Choi \cite{choilinear} asserts that any completely positive map can be characterized in terms of its \emph{Kraus operators}: $\{K_i\}_{i=1}^p$ are $m\times n$ matrices such that 
\begin{equation} \Phi(X) = \sum_{i=1}^p K_i X K_i^*\end{equation} for all $X\in M_n(\mathbb{C})$. The Kraus operators of a completely positive map are not necessarily unique. \\
In order to state Choi's theorem, we need to define the Choi matrix of a completely positive map.
\begin{defn} Let $\Phi:M_n(\C)\rightarrow M_m(\C)$ be a completely positive map. The Choi matrix of $\Phi$, $C_{\Phi}\in M_{nm}(\C)$ is defined by 
$$C_{\Phi} = \sum_{i,j=1}^n E_{ij}\otimes \Phi(E_{ij}).$$
\end{defn}
\begin{thm}[Choi] A map $\Phi: M_n(\C)\rightarrow M_m(\C)$ is completely positive if and only if its Choi matrix is positive semidefinite.\\
Moreover, if we let $\hat{v}_i \in \mathbb{C}^{nm}$ be a set of orthonormal eigenvectors for $C_{\Phi}$ with associated eigenvalues $\lambda_i\neq 0$, and define $$v_i = \sqrt{\lambda_i}\hat{v}_i = \bigoplus_{j=1}^m v_{ij}$$ where $v_{ij}\in \mathbb{C}^n$, then a set of Kraus operators for $\Phi$ is $\{K_i\}_{i=1}^p$ where 
$$K_i = \sum_{j=1}^m v_{ij}e_j^*.$$ 
\end{thm}
See \cite{choilinear} for Choi's proof of this. \\
The relationship between the eigenvectors of $C_{\Phi}$ and the Kraus operators of $\Phi$ is an instance of the Choi-Jamiolkowski isomorphism:\\
\begin{defn} Let $A,B\in M_n(\C)$ and let $a,b$ be their representations as vectors in the space $\mathbb{C}^{n^2}\simeq M_n(\C)$, using the standard ordered basis $\{E_{ij}\}_{i=1,j=1}^{n \ \ \ n}$; that is, their vectorizations: $a=\mathrm{vec}(A)$, $b=\mathrm{vec}(B)$. The \emph{Choi-Jamiolkowski isomorphism} is the linear involution on $M_{n^2}$ sending $\overline{B}\otimes A$ to $ab^*$.
\end{defn}
The fact that this map is indeed an involution is easy to see from the fact that $ab^*$ and $\overline{B}\otimes A$ clearly have the same entries, and thus must be related by some permutation of their entries. A simple computation shows that this permutation is in fact an involution.
The Choi-Jamiolkowski isomorphism interchanges between the representing map of $\Phi$, $T_{\Phi}$ (see Equation \ref{rep} in the next section), and the Choi matrix of $\Phi$, $C_{\Phi}$; and it interchanges between rank one matrices $ab^*$ and pure tensor products $\overline{B}\otimes A$. 
\\
Any completely positive map $\Phi: M_n(\C)\rightarrow M_m(\C)$ has an adjoint map $\Phi^{\dagger}:M_m(\C)\rightarrow M_n(\C)$, defined through the trace inner product:
\begin{equation} \mathrm{Tr}(\Phi(X)^*Y) = \mathrm{Tr}(X^*\Phi^{\dagger}(Y)).\end{equation}
Writing $\Phi(X) = \sum_{i=1}^p K_iXK_i^*$ it is easy to see that \begin{equation} \Phi^{\dagger}(Y) = \sum_{i=1}^p K_i^*YK_i.\end{equation}

\subsection{Complements of CP Maps}
A completely positive map is \emph{trace-preserving} if 
\begin{equation} \mathrm{Tr}(\Phi(X)) = \mathrm{Tr}(X) \end{equation} for all $X\in M_n(\mathbb{C})$. This is easily seen to be equivalent to the condition that 
\begin{equation} \sum_{i=1}^p K_i^*K_i = I_n.\end{equation}
Notice that such a map is necessarily the adjoint of a unital completely-positive map: one for which
\begin{equation} \Phi^{\dagger}(I) = \sum_{i=1}^p K_i^*K_i = I.\end{equation}

The Stinespring dilation theorem, \cite{stinespring}, characterizes trace-preserving completely positive maps as follows:

\begin{thm}[Stinespring] If $\Phi:M_n(\mathbb{C})\rightarrow M_m(\mathbb{C})$ is a trace-preserving completely positive map, there exists an integer $p \leq nm$ and a unitary $U$ on $\C^n\otimes \C^p$ such that
\begin{equation} \Phi(X) = (\mathrm{Id}\otimes \mathrm{Tr})\bigl( U(X\otimes E_{11})U^*\bigr).\end{equation}
\end{thm}

The theorem follows from the fact that $$\sum_{i=1}^p K_i \otimes e_i,$$ where $\{K_i\}_{i=1}^p$ are Kraus operators for $\Phi$, is an isometry, and so can be completed to a unitary $$U = \sum_{i=1}^p K_i\otimes E_{i1} +\sum_{i=1,j=2}^p U_{ij}\otimes E_{ij}$$ on $\C^n\otimes \C^p$. Then 

\begin{align} U(X\otimes E_{11})U^* & = \sum_{i=1}^p K_i X K_j^* \otimes E_{i1}E_{1j} \\
&\label{Stines} = \sum_{ij} K_i X K_j^*\otimes E_{ij}.\end{align}
Tracing out the second system results in $\Phi(X)$. 

\begin{defn} Given a trace-preserving completely positive map $\Phi:M_n(\mathbb{C})\rightarrow M_m(\mathbb{C})$ with Stinespring dilation
$$ \Phi(X) = (\mathrm{Id}\otimes \mathrm{Tr})\bigl( U(X\otimes E_{11})U^*\bigr),$$
the \emph{complementary channel} $\Phi^C$ is the channel defined by 
$$\Phi^C(X) = (\mathrm{Tr}\otimes \mathrm{Id})\bigl(U(X\otimes E_{11})U^*\bigr).$$
Clearly $\Phi^C$ has domain $M_n(\mathbb{C})$ and codomain $M_p(\mathbb{C})$. 
\end{defn}
From Equation \ref{Stines}, it is easy to see that 
\begin{equation}\label{complement} \Phi^C(X) = \sum_{i,j=1}^p \mathrm{Tr}(K_j^*K_i X)E_{ij}.\end{equation}
It is easy to see that if $\Phi$ is trace-preserving, so is $\Phi^C$:
\begin{align} \mathrm{Tr}(\Phi^C(X)) & = \sum_{i,j=1}^p \mathrm{Tr}(K_j^*K_i X)\mathrm{Tr}(E_{ij}) \\
& = \sum_{i=1}^p \mathrm{Tr}(K_i^*K_i X) \\
& = \mathrm{Tr}( \sum_{i=1}^p K_i^*K_i X) \\
& = \mathrm{Tr}(X).\end{align}
Our notion of complementarity is the one used in \cite{kks}, deriving ultimately from \cite{devshor}.

Even if $\Phi$ is not trace-preserving, and thus not subject to the Stinespring theorem, we may still use Equation \ref{complement} to define a complementary channel.

\begin{prop}\label{opsofcomp} Let $\Phi$ be a completely positive map $M_n(\C)\rightarrow M_m(\C)$, with Kraus operators $\{K_i\}_{i=1}^p$. Let $K_{ij}$, $1\leq j\leq n$, $1\leq i\leq p$ be the $j^{th}$ row of $K_i$, so that 
$$K_i = \sum_{j=1}^m e_jK_{ij}.$$
Then the Kraus operators of $\Phi^C$ can be chosen to be $\{L_i\}_{i=1}^n$ given by 

$$L_i = \sum_{j=1}^p e_jK_{ji}.$$
\end{prop}
\begin{proof} This is easiest to see from the Stinespring dilation. Let $\Phi$ have Stinespring dilation as in \ref{Stines}. Let $S:\C^n\otimes \C^k$ be the unitary that interchanges tensor factors:
\begin{equation} S(v\otimes w)=w\otimes v\end{equation}
for $v\in \mathbb{C}^n$, $w\in \C^k$. Then 
\begin{align} \Phi^C(X) & = (\mathrm{Tr}\otimes \mathrm{Id})\bigl(U(X\otimes E_{11})U^*\bigr) \\
& = (\mathrm{Id}\otimes \mathrm{Tr})\bigl(SU(X\otimes E_{11})U^*S^*\bigr)\end{align}
so a set of Kraus operators for $\Phi^C$ can be found by looking at the first block-column of the unitary matrix $SU$. This first row of $U$ is 
$$\sum_{i=1}^p K_i \otimes e_i = \sum_{i=1}^p\sum_{j=1}^n e_jK_{ij}\otimes e_i$$ so the first row of $SU$ is 
$$S\sum_{i=1}^p\sum_{j=1}^n e_jK_{ij}\otimes e_i = \sum_{i=1}^p\sum_{j=1}^n e_iK_{ij}\otimes e_j$$ and hence $L_j = \sum_{i=1}^p e_iK_{ji}$ for $1\leq j \leq n$.
\end{proof}

\subsection{The Operator System of a Completely Positive Map}
\begin{defn} An operator system is a subspace $S\subseteq M_n(\C)$ that contains the identity and is closed under $*$.
\end{defn}
The book \cite{paulsen} is a good reference for more on operator spaces. 
To each trace-preserving completely positive map $\Phi:M_n(\C)\rightarrow M_m(\C)$, we can assign an operator system $S_{\Phi}$ as follows:
\begin{equation} S_{\Phi} = \mathrm{span}\{ K_i^*K_j: \ \Phi(X) = \sum_{i=1}^p K_i X K_i^*\}.\end{equation}

Let $\Phi^C$ be the complementary channel to $\Phi$, and let $\Phi^{C\dagger}$ be the adjoint of the complement. Then
\begin{align}\mathrm{Tr}(\Phi^{C\dagger}(E_{ij})^*X) & = \mathrm{Tr}(E_{ji}\Phi^C(X))\\
& = \mathrm{Tr}(K_i^*K_jX)\end{align} for all $X\in M_n(\C)$, $1\leq i,j\leq m$, where we used Equation \ref{complement} to go to the second line. Hence 
\begin{equation}\Phi^C(E_{ij})=K_i^*K_j\end{equation}
and so $S_{\Phi}=\mathrm{span}\{K_i^*K_j\} = \mathrm{range}(\Phi^{C\dagger})$. 

\subsection{Uncertainty Principle for CP Maps}
Since a completely positive map $\Phi:M_n(\mathbb{C})\rightarrow M_m(\mathbb{C})$ is a linear map on the vector space of $n\times n$ matrices, it can be represented by an $n^2\times n^2$ matrix acting on the vector space $M_n(\mathbb{C})$ with ordered basis $\{E_{ij}\}_{i=1,j=1}^{n \ \ \ n}$. It is an easy computation that multiplication of the standard basis for $M_n(\mathbb{C})$ by the matrix $A$ has representation $L_A = I\otimes A$, and multiplication on the right is represented by $R_A = A^T\otimes I$. Then the representing matrix of $\Phi$, $T_{\Phi}$ is given by 
\begin{equation}\label{rep} T_{\Phi} = \sum_{i=1}^p \overline{K}_i \otimes K_i.\end{equation}

We will make use of the following lemma in our proof of the main theorem:
\begin{lem}\label{singvals} For a matrix $A\in M_n(\mathbb{C})$ with singular values $\|A\|=\sigma_1\geq \sigma_2 \geq \ldots \geq \sigma_n$,
\begin{equation} \|A\|^2 \mathrm{rank}(A) \geq \mathrm{Tr}(AA^*).\end{equation}
\end{lem}
\begin{proof} \begin{align*} \|A\|^2 \mathrm{rank}(A) & = \sigma_1^2 \mathrm{rank}(A) \\
& \geq \sum_{i=1}^{\mathrm{rank}(A)} \sigma_i^2 \\
& = \sum_{i=1}^n \sigma_i^2 \\
& = \mathrm{Tr}(AA^*).\end{align*}
\end{proof}

\begin{lem}\label{normcp} Let $\Phi:M_n(\mathbb{C})\rightarrow M_m(\mathbb{C})$ be a completely positive map and $T_{\Phi} \in \mathcal{L}(\C^{n\times n},\C^{m\times m})$ be its representing matrix. Then 
\begin{equation} \mathrm{Tr}(T_{\Phi}T_{\Phi}^*) = \|\Phi^C(I)\|_{HS}^2\end{equation}
where the norm on the right side is the Hilbert Schmidt norm.
\end{lem}
\begin{proof} If $\Phi$ has Kraus operators $K_i$, then by Equation \ref{rep} we have that 
$$T_{\Phi}= \sum_{i=1}^p \overline{K}_i \otimes K_i$$ and so
$$T_{\Phi}T_{\Phi}^* = \sum_{i=1}^p\sum_{j=1}^p \overline{K}_iK_j^T \otimes K_i K_j^*.$$

Taking the trace yields 
\begin{align} \mathrm{Tr}(T_{\Phi}T_{\Phi}^*) & = \sum_{i=1}^p\sum_{j=1}^p \mathrm{Tr}(\overline{K}_iK_j^T \otimes K_i K_j^*) \\
& = \sum_{i=1}^p\sum_{j=1}^p \mathrm{Tr}(\overline{K}_iK_j^T)\mathrm{Tr}(K_i K_j^*) \\
&\label{sumtrace} = \sum_{i=1}^p\sum_{j=1}^p \bigl| \mathrm{Tr}(K_iK_j^*) \bigr|^2.\end{align}

Recalling Equation \ref{complement}, 
\begin{equation} \Phi^C(I) = \sum_{i=1}^p\sum_{j=1}^p \mathrm{Tr}(K_i^*K_j)E_{ji}\end{equation} 
and so \begin{equation}\label{compofid}\|\Phi^C(I)\|_{HS}^2 = \sum_{i=1}^p\sum_{j=1}^p \bigl|\mathrm{Tr}(K_i^*K_j)\bigr|^2.\end{equation}

Comparing Equation \ref{sumtrace} to Equation \ref{compofid} yields the desired equality.
\end{proof}

We are now in a position to prove our main theorem:
\begin{thm}\label{main} Let $\Phi:M_n(\mathbb{C})\rightarrow M_m(\mathbb{C})$ with Kraus operators $\{K_i\}_{i=1}^p$ and complementary channel $\Phi^C$. Then
\begin{equation} \mathrm{rank}(\Phi)\mathrm{rank}(\Phi^C)\geq \frac{\|\Phi(I)\|_{HS}^2\|\Phi^C(I)\|_{HS}^2}{\|\Phi\|^2\|\Phi^C\|^2}.\end{equation}
\end{thm}
\begin{proof} This is an easy consequence of Lemmas \ref{singvals} and \ref{normcp} applied to the matrix $T_{\Phi}$. By Lemma \ref{singvals}, 
$$\mathrm{rank}(\Phi) \geq \frac{\mathrm{Tr}(T_{\Phi}T_{\Phi}^*)}{\|\Phi\|^2}$$ and using Lemma \ref{normcp} to substitute in for $\mathrm{Tr}(T_{\Phi}T_{\Phi}^*)$ gives 
\begin{equation}\label{1}\mathrm{rank}(\Phi)\geq \frac{\|\Phi^C(I)\|_{HS}^2}{\|\Phi\|^2}.\end{equation}
Doing the same for $\Phi^C$ gives 
\begin{equation}\label{2} \mathrm{rank}(\Phi^C) \geq \frac{\|\Phi(I)\|_{HS}^2}{\|\Phi^C\|^2}\end{equation} and multiplying Inequalities \ref{1} and \ref{2} together gives the result.
\end{proof} 

\section{Unital Channels and Examples} 
An important class of completely positive maps are the unital maps, those satisfying
\begin{equation} \Phi(I) = \sum_{i=1}^p K_iK_i^* = I.\end{equation}

If a completely positive map $\Phi:M_n(\C)\rightarrow M_m(\C)$ is both trace-preserving and unital, then $\Phi$ is necessarily a contraction. The proof, which we present for the sake of completeness, can be found in \cite{paulsen}.
\begin{proof} Let $\Phi$ be unital, completely positive, and trace-preserving. Let 
$$A = \begin{bmatrix} I & 0 \\ a & 0 \end{bmatrix}\begin{bmatrix} I & a^* \\ 0 & 0 \end{bmatrix} = \begin{bmatrix} I & a^* \\ a & aa^*\end{bmatrix},$$ for $a\in M_n(\mathbb{C})$. 
Clearly $A$ is positive, and since $\Phi$ is completely positive, 
$$(\mathrm{Id_2}\otimes \Phi)(A) = \begin{bmatrix} \Phi(I) & \Phi(a^*) \\ \Phi(a) & \Phi(aa^*)\end{bmatrix} = \begin{bmatrix} I & \Phi(a)^* \\ \Phi(a) & \Phi(aa^*)\end{bmatrix}$$ is positive as well. A Schur complement argument shows that \begin{equation}\label{contract}\Phi(aa^*) \geq \Phi(a)\Phi(a)^*.\end{equation}
Taking traces of Equation \ref{contract} and using the fact that $\Phi$ is trace-preserving 
\begin{align} \|a\|_{HS}^2 & = \mathrm{Tr}(aa^*) \\
& = \mathrm{Tr}(\Phi(aa^*)) \\
& \geq  \mathrm{Tr}(\Phi(a)\Phi(a)^*) \\
& = \|\Phi(a)\|_{HS}^2.\end{align}
\end{proof}

\begin{lem}\label{unitallem} If $\Phi:M_n(\C)\rightarrow M_m(\C)$ is a unital CP map, then $\Phi^C$ satisfies 
\begin{equation} \|\Phi^C\|_{HS}^2 \leq m.\end{equation}
\end{lem}
\begin{proof} This is a straightforward calculation:
\begin{align} \|\Phi^C(X)\|_{HS}^2 & = \sum_{i,j=1}^p |\mathrm{Tr}(K_i^*K_jX)|^2 \\
& \leq \sum_{i,j=1}^p \|K_i^*K_j\|_{HS}^2\|X\|_{HS}^2\end{align}
and so 
\begin{align} \|\Phi^C\|_{HS}^2 & \leq \sum_{i,j=1}^n \mathrm{Tr}(K_i^*K_jK_j^*K_i) \\
& = \mathrm{Tr}(\sum_{i,j=1}^p K_iK_i^*K_jK_j^*)\\
& = \mathrm{Tr}(I)\\
& = m.\end{align}
\end{proof}

\begin{lem}\label{tplem} If $\Phi:M_n(\C)\rightarrow M_m(\C)$ is trace-preserving and CP, then $\Phi^C$ satisfies
\begin{equation} \|\Phi^C(I)\|_{HS}^2 \geq n^2. \end{equation}
\end{lem}
\begin{proof} Again, this is straightforward:
\begin{align} \|\Phi^C(I)\|_{HS}^2 &= \sum_{i,j=1}^p \bigl|\mathrm{Tr}(K_i^*K_j)\bigr|^2 \\
& \geq \bigl|\sum_{i\neq j} \mathrm{Tr}(K_i^*K_j) \bigr|^2 + \bigl|\sum_{i=1}^p \mathrm{Tr}(K_i^*K_i) \bigr|^2 \\
& \geq \bigl| \mathrm{Tr}(I)\bigr|^2 \\
& = n^2.\end{align}
\end{proof}

\begin{cor} If $\Phi:M_n(\mathbb{C})\rightarrow M_n(\mathbb{C})$ is a completely positive map that is both unital and trace-preserving then 
\begin{equation} \mathrm{Rank}(\Phi)\mathrm{Rank}(\Phi^C) \geq n^2.\end{equation}
\end{cor}
\begin{proof} The proof is simply a combination of the fact that $\|\Phi(I)\|_{HS}^2 = n$ if $\Phi$ is unital, the fact that a unital and trace-preserving map is a contraction, and Lemmas \ref{unitallem} and \ref{tplem}. 
\end{proof}
We now present some interesting examples of channels satisfying this.
\subsection{Unitary Adjunction Channels}
Let $\Phi:M_n(\mathbb{C})\rightarrow M_n(\mathbb{C})$ be a unitary adjuntion channel:
$$\Phi(X) = UXU^*$$ for some unitary matrix $U$. Clearly $\Phi$ is invertible, so $\mathrm{rank}(\Phi) = n^2$. Further, $\Phi^C(X) = \mathrm{Tr}(X)$, and obviously has $\mathrm{rank}(\Phi^C)=1$. In this case, 
\begin{equation} \mathrm{rank}(\Phi)\mathrm{rank}(\Phi^C) = n^2.\end{equation}
This is a simple example of the case where equality is achieved.
\subsection{Rank One Channels}
Let $\Phi:M_n(\C)\rightarrow M_m(\C)$ be trace-preserving and rank-one: $\Phi(X)=\mathrm{Tr}(X)\rho_0$ for some fixed $\rho_0$ with trace $1$. The Choi matrix for such a channel is 
$$C_{\Phi} = I\otimes \rho_0$$ and is positive if and only if $\rho_0\geq 0$. Let $\{v_i\}_{i=1}^{\mathrm{rank}(\rho_0)}$ be a set of orthonormal eigenvectors corresponding to non-zero eigenvalues of $\rho_0$. Then a set of eigenvectors of $C_{\Phi}$ is $\{e_i\otimes v_j\}_{i=1,j=1}^{n,\ \ \ \mathrm{rank}(\rho_0)}$. \\
A set of Kraus operators for this channel is then 
\begin{equation}\{K_{ij}\}_{i=1,j=1}^{n,\ \ \ \mathrm{rank}(\rho_0)}=\{\sqrt{\lambda_j}v_je_i^*\}_{i=1,j=1}^{n,\ \ \ \mathrm{rank}(\rho_0)}.\end{equation}
The rank of $\Phi^C$ is equal to the rank of $\Phi^{C\dagger}$ which is in turn the dimension of the operator system $$S_{\Phi} =\mathrm{span}\{K_{ij}^*K_{kl}\}.$$
Fix $j$, and consider 
\begin{equation} K_{ij}^*K_{kj} = \sqrt{\lambda_j}^2 (v_je_i^*)^*(v_je_k^*)  = \lambda_j E_{ik}.\end{equation}
Varying $1\leq i,k\leq n$ we find that $S_{\Phi} = M_n(\C)$ and so $\mathrm{rank}(\Phi^C) = n^2$. Hence
\begin{equation} \mathrm{rank}(\Phi)\mathrm{rank}(\Phi^C) = n^2.\end{equation}
Again, in this case, equality is achieved.
\subsection{Schur Product Channels}
Let $\Phi:M_n(\C)\rightarrow M_n(\C)$ be a Schur product channel:
$$\Phi(X) = C\circ X$$ where $\circ$ denotes the entry-wise Schur product of two matrices, and $C$ is a positive semidefinite matrix. $\Phi$ is trace-preserving if and only if $C$ has all $1$s down its diagonal, in which case $\Phi$ is also automatically unital. A positive semidefinite matrix with $1$s down the diagonal is called a correlation matrix.\\
Denote an orthonormal set of eigenvectors for $C$ by $\hat{v}_i$, so that 
$$ C = \sum_{i=1}^{\mathrm{rank}(C)} \lambda_i \hat{v}_i\hat{v}_i^* = \sum_{i=1}^{\mathrm{rank}(C)} v_iv_i^*,$$
where $v_i = \sqrt{\lambda_i}\hat{v}_i$.\\
Then it is well-known that the Kraus operators for $\Phi$ can be chosen to be diagonal matrices, $$K_i = \mathrm{diag}(v_i).$$

If $V$ is the matrix whose $i^{th}$ column is $v_i$, let $w_i\in\mathbb{C}^{m}$ be the vector whose adjoint is the $i^{th}$ row of $V$, for $1\leq i\leq n$. Then $VV^*=C$ and $C_{ij}=w_i^*w_j$.\\

By Proposition \ref{opsofcomp} the Kraus operators for $\Phi^C$ are $L_i = \overline{w}_ie_i^*$, and so the Kraus operators for $\Phi^{C\dagger}$ are $L_i^*=e_iw_i^T$. \\

The matrix representing the channel $\Phi^{C\dagger}$ is then 
\begin{equation} T_{\Phi^{C\dagger}} = \sum_{i=1}^n e_iw_i^*\otimes e_iw_i^T = \sum_{i=1}^n (e_i\otimes e_i)(w_i\otimes \overline{w}_i)^*.\end{equation}
Ignoring the rows with only zero entries, the rank of $T_{\Phi^{C\dagger}}$, and hence of $\Phi^{C\dagger}$, is equal to the rank of $T = \sum_{i=1}^n e_i(w_i\otimes \overline{w}_i)^*$.\\
However, $$TT^* = \sum_{i,j=1}^n e_ie_j^* (w_i^*w_j)(w_i^T\overline{w}_j) = \sum_{i,j=1}^n E_{ij} C_{ij}\overline{C}_{ij}.$$

Hence $TT^* = C\circ\overline{C}$, and so 
\begin{equation}\mathrm{rank}(T)=\mathrm{rank}(TT^*)=\mathrm{rank}(C\circ\overline{C}).\end{equation}

Observe that the rank of $\Phi$ is equal to the number of non-zero entries in $C$. This is of course the same as the number of non-zero entries of $C\circ\overline{C}$. Denote this number by $N(C\circ\overline{C})$.\\
$C\circ \overline{C}$ is a positive semidefinite matrix all of whose entries are nonnegative and real, so $$\mathrm{Tr}((C\circ\overline{C})^2) = \|C\circ\overline{C}\|_{HS}^2 = \sum_{i,j=1}^n a_{ij}^2 \leq N(C\circ\overline{C}).$$

\begin{lem}\label{ranktraces} Let $A$ be a positive semidefinite matrix.

$$\mathrm{Tr}(A)^2 \leq \mathrm{rank}(A)\mathrm{Tr}(A^2).$$
\end{lem}
See \cite{deaett}\cite{pudlak} for other uses of this lemma in relating the rank to the structure of a positive semidefinite matrix. For convenience, we provide a simple proof, from \cite{deaett}.
\begin{proof} Suppose $A$ has rank $k$, with non-zero eigenvalues $\{\lambda_i\}_{i=1}^k$. \\
Let $\Lambda = (\lambda_1,\lambda_2,\ldots,\lambda_k)$ and let $K = \frac{1}{\sqrt{k}}(1,1,\ldots,1)$. By the Cauchy-Schwarz inequality
\begin{align} |\langle \Lambda,K\rangle|^2 &= \frac{1}{k}(\mathrm{Tr}(A))^2\\
& \leq \|\lambda\|^2\|K\|^2\\
& = \mathrm{Tr}(A^2).\end{align}  
\end{proof}

Since $C$ is a correlation matrix, $\mathrm{Tr}(C\circ \overline{C})=n$, so invoking Lemma \ref{ranktraces} we have 
\begin{align} n^2 &= \mathrm{Tr}(C\circ \overline{C})^2\\
&\leq \mathrm{rank}(C\circ\overline{C})\mathrm{Tr}((C\circ\overline{C})^2) \\
& \leq \mathrm{rank}(C\circ\overline{C})N(C\circ\overline{C}).\end{align}

Finally, we appeal to the fact that $N(C\circ\overline{C})=\mathrm{rank}(\Phi)$ and $$\mathrm{rank}(C\circ \overline{C})=\mathrm{rank}(\Phi^{C\dagger})=\mathrm{rank}(\Phi^C)$$ to show that
\begin{equation} \mathrm{rank}(\Phi)\mathrm{rank}(\Phi^C)\geq n^2.\end{equation} 

Although equality is not always achieved in the Schur product case, we can characterize the cases for which it is the case. This will occur whenever all eigenvalues of $C\circ\overline{C}$ are either $0$ or $\frac{n}{\mathrm{rank}(C)}$, and all entries of $C\circ\overline{C}$ (and hence $C$ itself) are modulus $0$ or $1$.\\
The latter occurs whenever $C\circ \overline{C}$ is a direct sum of $\mathrm{Rank}(C\circ\overline{C})$ rank-$1$ correlation matrices, each of which is necessarily of the form $v_iv_i^*$ for some $v_i\in \C^{k_i}$ with all entries of modulus $1$ (see Proposition \ref{onezerosofC} in Section $4$). \\
The eigenvalues of such a matrix are $\{k_i\}$, the dimensions of the direct summands. If they are all to be equal then $k_i=k$ for some fixed $k$ and all $i$, so that $C\circ\overline{C}$ is the direct sum of $\mathrm{Rank}(C\circ\overline{C})=\frac{n}{k}$ rank-$1$ correlation matrices each of dimension $k\times k$. 
\section{Doubly Stochastic Matrices}
It is well-known that unital, trace-preserving completely positive maps are in some ways analogous to doubly stochastic matrices.
Recall that a doubly stochastic (DS) matrix $D$ is a real $n\times n$ matrix such that $d_{ij}\geq 0$ for all $i,j\in \langle n\rangle$ and
\begin{equation}\label{defn} \sum_{i=1}^n d_{ij} = \sum_{j=1}^n d_{ij} = 1 \ \ \ \forall \ i,j \in \langle n\rangle.\end{equation}

Note that Equation \ref{defn} immediately implies that \begin{equation}\label{constraint}\sum_{i=1}^n \sum_{j=1}^n d_{ij} = n.\end{equation} 

There are many interesting connections between doubly stochastic matrices and unital, trace-preserving completely positive maps. In fact, completely positive maps that are both unital and trace-preserving are sometimes called doubly stochastic. If $e$ is the all $1$s vector in $\mathbb{R}^n$, then a matrix $D\in M_n(\mathbb{R})$ with positive entries is doubly stochastic if and only if 
$$De = e \textnormal{ and } e^*D = e^*.$$
Analogously, if $T_{\Phi}$ is the representing matrix for a completely positive map $\Phi$, and $E$ is the vectorization of the identity matrix, $\Phi$ is unital and trace-preserving if and only if 
$$T_{\Phi}E = E \textnormal{ and } E^*T_{\Phi} = E^*.$$

It has also been observed by Chru{\'s}ci{\'n}ski in \cite{chruscinski} that if $\Phi$ is a trace-preserving and unital completely positive map, the matrix $D_{\Phi}$ defined by
\begin{equation} D_{\Phi_{ij}} = \Phi(E_{ii})_{jj}\end{equation} is doubly stochastic. If $\Phi$ has Kraus operators $\{K_i\}_{i=1}^p$, then 
$D_{\Phi_{ij}} = \sum_{k=1}^p e_j^*K_k e_ie_i^* K_k^*e_j = \sum_{k=1}^p |(K_k)_{ij}|^2$ and hence 
\begin{equation} D_{\Phi} = \sum_{i=1}^p K_i\circ \overline{K}_i.\end{equation}

Notice that the $i^{th}$ column of $D_{\Phi}$ is the diagonal of $\Phi(E_{ii})$, and so the columns of $D_{\Phi}$ appear down the diagonal of the Choi matrix, $C_{\Phi}$. \\
Properties of this doubly stochastic matrix may sometimes be related to properties of $\Phi$. 
\begin{thm}\label{rankds} Let $\Phi: M_n(\C) \rightarrow M_m(\C)$ be a unital, trace-preserving, completely positive map, with associated doubly stochastic matrix $D_{\Phi}$. Then
\begin{equation} \mathrm{rank}(\Phi)\geq \mathrm{rank}(D_{\Phi}).\end{equation}
\end{thm}

\begin{proof} Let $T_{\Phi}$ be the matrix representation of $\Phi$, with singular value decomposition
$$T_{\Phi} = \sum_{i=1}^{\mathrm{rank}(\Phi)} \sigma_i u_i v_i^*.$$
By the Choi-Jamiolkowski isomorphism, this implies that the Choi matrix $C_{\Phi}$ can be written as 
$$C_{\Phi} = \sum_{i=1}^{\mathrm{rank}(\Phi)} \sigma_i \overline{V}_i\otimes U_i$$ where $U_i$ and $V_i$ are matrices whose representation as vectors in the vector space $\C^{n\times n}\simeq M_n(\C)$ are $u_i$ and $v_i$ respectively. Let $\hat{u}_i$ be the diagonal of $U_i$ and $\hat{v}_i$ be the diagonal of $V_i$. Then the diagonal of $C_{\Phi}$ is equal to $\sum_{i=1}^{\mathrm{rank}(\Phi)} \sigma_i \hat{v}_i\otimes \hat{u}_i$. \\
Making use of the isomorphism $v\otimes w \leftrightarrow vw^*$, we obtain that
\begin{equation} D_{\Phi} = \sum_{i=1}^{\mathrm{rank}(\Phi)} \sigma_i \hat{v}_i\hat{u}_i^*\end{equation}
and so clearly the rank of $D_{\Phi}$ can be no bigger than the rank of $\Phi$. 
\end{proof}

In light of this relationship between unital, trace-preserving completely positive maps and doubly stochastic matrices, and the relation between their ranks, as in Theorem \ref{rankds}, it is interesting to note that we can obtain the following bound on the rank of a doubly stochastic matrix. 
\begin{thm} Let $D$ be an $n\times n$ doubly stochastic matrix, with rank $r$ and number of non-zero entries $N$. Then
\begin{equation} Nr\geq n^2.\end{equation}
\end{thm}

\begin{proof} By the Perron-Frobenius theorem, $D$ is a contraction, and hence for every singular value $\sigma_i$ of $D$ we have 
\begin{equation} \sigma_i \leq 1.\end{equation} See, for example, \cite{minc}.
Then it is clear that 
\begin{align}\label{rank}\mathrm{Tr}(DD^T) &= \sum_{i=1}^n \sigma_i^2 &\\ &= \sum_{i=1}^r \sigma_i^2 \\
&\leq r.\end{align}

We now seek to minimize $\mathrm{Tr}(DD^T) = \sum_{i=1}^n\sum_{j=1}^n d_{ij}^2$ subject to the constraint   
in Equation \ref{constraint}. \\
If we define the set $Z_D: = \{ (i,j): i,j\in \langle n\rangle, \ d_{ij} = 0 \} \subseteq \langle n\rangle\times \langle n\rangle$, with set theoretic complement $\overline{Z}$ in $\langle n\rangle \times \langle n\rangle$ then $|\overline{Z}_D|=N$, and the expression we seek to optimize becomes
\begin{equation}\sum_{i=1}^n\sum_{j=1}^n d_{ij}^2 = \sum_{(i,j)\in \overline{Z}_D} d_{ij}^2\end{equation} and we minimize this subject to
\begin{equation} \sum_{(i,j)\in \overline{Z}_D} d_{ij} = n.\end{equation}

Using, for example, Lagrange multipliers, it is clear that the minimum occurs when all non-zero $d_{ij}$ are equal: 
\begin{equation} d_{ij} = \frac{n}{N}, \ (i,j)\in \overline{Z}_D.\end{equation}

Then \begin{equation}\label{Ns} \frac{n^2}{N} = \sum_{(i,j)\in \overline{Z}_D} \bigl(\frac{n}{N}\bigr)^2  \leq \sum_{(i,j)\in \overline{Z}_D} d_{ij}^2\end{equation}
and combining Equations \ref{Ns} and \ref{rank} we obtain
\begin{equation} n^2 \leq Nr.\end{equation}
\end{proof}

\section{Quantum Information}
Completely positive maps have been extensively studied in part because of their connection to the theory of quantum information. A quantum channel is a map that preserves quantum states; mathematically, quantum channels are trace-preserving completely positive maps, an idea going back to \cite{kraus}.\\
A quantum channel is said to be private, or to privatize some input set $S$ if there exists a fixed output $\rho_0$ such that 
\begin{equation} \Phi(X) = \frac{\mathrm{Tr}(X)}{\mathrm{Tr}(I)}\rho_0\end{equation}
for all $X\in S$. Usually, $\rho_0$ is chosen to be the identity, and $S$ is a set with the structure of a $*$-algebra, ideally one isomorphic to $M_{2^k}(\C)$ for some $k\in \mathbb{N}$--in this case, $\Phi$ privatizes $k$ qubits of information. The idea of private channels originates with \cite{ambainis} and has been generalized and further studied in \cite{boykin}\cite{devetak}  \cite{jochym}  \cite{kks}and \cite{levick}.  \\
A channel is correctable on a subset $S$ if $\Phi$ is invertible for all inputs $X\in S$. See \cite{kribslaflamme} or \cite{kribsspekkens} for more on correctable channels.\\
In \cite{kks} it was shown that for certain kinds of privacy, there is a trade-off between the amount of information a channel and its complement can privatize: the degree to which $\Phi$ is private is the degree to which $\Phi^C$ is correctable. The relationship between complementarity, privacy, and correctability was studied further in \cite{plosker}. \\

Our result can be regarded as a companion to this result, in that it provides a similar floor to the degree to which a channel and its complement can both be private; albeit with less to say about the structure of the subsets that are privatized. Our result shows that both $\Phi$ and $\Phi^C$ cannot both be ``highly non-invertible"--the lower the rank of $\Phi$, the greater the rank of $\Phi^C$. \\
\begin{thm} Let $\Phi:M_n(\C)\rightarrow M_m(\C)$ be a completely positive map, with adjoint $\Phi^{\dagger}$. Let $S$ be the operator system $\mathrm{range}(\Phi^{\dagger})$. An algebra $\mathcal{A}$ is privatized by $\Phi$ if and only if 
\begin{equation}\label{priv} \mathrm{Tr}(AX) = \frac{\mathrm{Tr}(A)\mathrm{Tr}(X)}{n}\end{equation}
for all $A\in \mathcal{A}$, $X\in S$. 
\end{thm}
\begin{proof} $S = \Phi^{\dagger}(M_n(\C))$, so assume Equation \ref{priv} holds for some algebra $\mathcal{A}$, then 
\begin{align} \mathrm{Tr}(Y\Phi(A)) & =  \mathrm{Tr}(\Phi^{\dagger}(Y)A) \\
& = \frac{\mathrm{Tr}(\Phi^{\dagger}(Y))\mathrm{Tr}(A)}{n}\\
& = \frac{\mathrm{Tr}(Y\Phi(I))\mathrm{Tr}(A)}{n}\\
& = \mathrm{Tr}(Y\frac{\mathrm{Tr}(A)\Phi(I)}{n})\end{align}
for all $A \in \mathcal{A}$ and $Y\in M_m(\C)$ so $\Phi(A) = \frac{\mathrm{Tr}(A)}{n}\Phi(I)$.
The other direction is proved by following the steps in reverse. 
\end{proof}
When this trace condition holds between two unital $*$-subalgebras of $M_n(\C)$, we say that the two algebras are quasiorthogonal. It is equivalent to the two algebras being completely orthogonal except for their common intersection in the subspace spanned by the identity. Quasiorthogonality underlies many interesting phenomena in quantum mechanics, including for example mutually unbiased bases. See \cite{petz} for more on this connection, and \cite{levick} for more on the relationship between private channels and quasiorthogonality.\\

Associated to any channel is an algebra, the multiplicative domain:
\begin{defn} Let $\Phi^{\dagger}:M_m(\C)\rightarrow M_m(\C)$ be a unital completely positive map. Then the multiplicative domain of $\Phi^{\dagger}$ is the set 
$$MD(\Phi^{\dagger}) = \{ X \in M_m(\C): \Phi^{\dagger}(X)\Phi^{\dagger}(Y) = \Phi^{\dagger}(XY) \ \forall Y \in M_m(\C)\}.$$
\end{defn}
See \cite{choischwarz} for basic facts on the multiplicative domain.

The multiplicative domain of a unital completely positive map is obviously an algebra. 
\begin{defn} Given a unital, trace-preserving completely positive map $\Phi:M_n(\C)\rightarrow M_n(\C)$, the fixed point algebra is the set 
$$\mathrm{Fix}(\Phi) = \{ X: \Phi(X) = X\}.$$
\end{defn}
If $\Phi$ is trace-preserving and unital, $\mathrm{Fix}(\Phi) = \{K_i\}_{i=1}^{p\prime}$, that is, the fixed point set is equal to the set of matrices that commute with each Kraus operator of $\Phi$, a fact proven in \cite{kribswavelets}. 

\begin{thm}\label{MDfix} Let $\Phi:M_n(\C)\rightarrow M_m(\C)$ be a trace-preserving, completely positive map. Then 
$$MD(\Phi)\subseteq \mathrm{Fix}(\Phi^{\dagger}\circ\Phi).$$
\end{thm}

\begin{proof} Let $X\in MD(\Phi)$, and $Y$ be arbitrary. Then
\begin{align} \mathrm{Tr}(XY) & = \mathrm{Tr}(\Phi(XY)) \\
& = \mathrm{Tr}(\Phi(X)\Phi(Y))\\
& = \mathrm{Tr}(\Phi^{\dagger}(\Phi(X))Y).\end{align}
Hence $X=(\Phi^{\dagger}\circ\Phi)(X)$.
\end{proof}

\begin{thm}\label{unitalmd} Let $\Phi:M_n(\C)\rightarrow M_m(\C)$ be a unital completely positive map with Kraus operators $\{K_i\}_{i=1}^p$, and hence operator system $S_{\Phi} = \mathrm{span}\{K_i^*K_j\}_{i,j=1}^p$. Then $S_{\Phi}^{\prime} \subseteq MD(\Phi)$.
\end{thm}
\begin{proof} Assume $[A,K_i^*K_j]=0$ for all $1\leq i,j\leq p$. Then 
\begin{align} \Phi(A)\Phi(X) & = \sum_{i,j=1}^p K_iAK_i^*K_jXK_j^* \\
& = \sum_{i,j=1}^p K_iK_i^*K_jAXK_j^* \\
& = \sum_{j=1}^p K_jAXK_j^* \\
& = \Phi(AX)\end{align} 
where we use the unitality of $\Phi$ to go from $(91)$ to $(92)$. 
\end{proof}

 \begin{defn} Given an operator system $S$, define $\mathrm{Alg}(S)$ to be the smallest $*$-subalgebra containing $S$.
\end{defn}
Since a $*$-subalgebra is its own double commutant, $\mathrm{Alg}(S)= S^{\prime\prime}$. 

\begin{thm}\label{MDtpunital} Let $\Phi:M_n(\C)\rightarrow M_m(\C)$ be a trace-preserving, unital completely positive map, with Kraus operators $\{K_i\}_{i=1}^p$ and operator system $S_{\Phi}=\mathrm{span}\{K_i^*K_j\}_{i,j=1}^p$. Then $\mathrm{Alg}(S_{\Phi}) = MD(\Phi)^{\prime}$. 
\end{thm}
\begin{proof} Since $\Phi$ is trace-preserving and unital, Theorems \ref{MDfix} and \ref{unitalmd} both apply, and so 
\begin{equation} S_{\Phi}^{\prime} \subseteq MD(\Phi)\subseteq \mathrm{Fix}(\Phi^{\dagger}\circ\Phi).\end{equation}
$\Phi^{\dagger}\circ\Phi$ is both unital and trace-preserving, and its Kraus operators are $\{K_i^*K_j\}_{i=1}^p$. The fixed point algebra of a unital and trace-preserving channel is the commutant of its Kraus operators, and so $\mathrm{Fix}(\Phi^{\dagger}\circ\Phi) = S_{\Phi}^{\prime}=MD(\Phi)$. 
Hence
\begin{equation} \mathrm{Alg}(S_{\Phi}) = S_{\Phi}^{\prime\prime} = MD(\Phi)^{\prime}.\end{equation}
\end{proof}

A sufficient condition for an algebra $\mathcal{A}$ to be privatized by a channel $\Phi$ is for $\mathcal{A}$ to be quasiorthogonal to the algebra generated by $\mathrm{Range}(\Phi^{\dagger})$. 
Hence, if $\Phi:M_n(\C)\rightarrow M_m(\C)$ is unital and trace-preserving, a sufficient condition for an algebra $\mathcal{A}$ to be privatized by $\Phi^C$ is that $\mathcal{A}$ be quasiorthogonal to $\mathrm{Alg}(S_{\Phi})=MD(\Phi)^{\prime}$. \\

\begin{thm}\label{phiphidagger} Let $\Phi:M_n(\C)\rightarrow M_m(\C)$ be a trace-preserving completely positive map. Let $MD(\Phi)$ be its multiplicative domain. Then $\Phi$ is a $*$-homomorphism when restricted to $MD(\Phi)$ that injects into $MD(\Phi^{\dagger})$.
\end{thm}
\begin{proof} That $\Phi$ is a $*$-homomorphism when restricted to $MD(\Phi)$ follows from the definition
of $MD(\Phi)$. The image of $MD(\Phi)$ under $\Phi$ is in $MD(\Phi^{\dagger})$ since for $A\in MD(\Phi)$, $X, Y$ arbitrary we have
\begin{align*} \mathrm{Tr}(\Phi^{\dagger}(\Phi(A)X)Y) & = \mathrm{Tr}(\Phi(A)X\Phi(Y)) \\
&= \mathrm{Tr}(\Phi(YA)X) \\
& = \mathrm{Tr}(A\Phi^{\dagger}(X)Y).\end{align*}
Thus, $\Phi^{\dagger}(\Phi(A)X) = A\Phi^{\dagger}(X) = \Phi^{\dagger}(\Phi(A))\Phi^{\dagger}(X)$ where the last equality comes from applying Theorem \ref{MDfix}.\\
Finally, the mapping from $MD(\Phi)$ to $MD(\Phi^{\dagger})$ is injective since Theorem \ref{MDfix} assures us that $\Phi^{\dagger}$ is an inverse for this map. 
\end{proof}

If $\Phi$ is trace-preserving, then $\Phi^{\dagger}$ is unital, so that 
$$\Phi(MD(\Phi))\subseteq MD(\Phi^{\dagger}).$$ After acting on both sides by $\Phi^{\dagger}$ and recalling that $\Phi^{\dagger}(\Phi(MD(\Phi)))= MD(\Phi)$, we have
\begin{equation}MD(\Phi)\subseteq \Phi^{\dagger}(MD(\Phi^{\dagger})).\end{equation}
Since $\Phi^{\dagger}$ is unital, $I\in MD(\Phi^{\dagger})$ and so $MD(\Phi^{\dagger})$ is a unital $*$-subalgebra, and $\Phi^{\dagger}$ when restricted to $MD(\Phi^{\dagger})$ is a $*$-homomorphism whose image is clearly also a $*$-subalgebra contained in $\mathrm{Range}(\Phi^{\dagger})$.\\
Applying all of the above to the map $\Phi^{C\dagger}$ when $\Phi:M_n(\C)\rightarrow M_m(\C)$ is unital and trace-preserving, we see that there are two algebras naturally associated to the operator system $S_{\Phi} = \mathrm{Range}(\Phi^{C\dagger})$. The first is the algebra generated by $S_{\Phi}$, $S_{\Phi}^{\prime\prime}$ which we have already seen provides us with a sufficient condition for an algebra $\mathcal{A}$ to be privatized by $\Phi^C$, and in this special case is the commutant of $MD(\Phi)$. \\
The second is the algebra $\Phi^{\dagger}(MD(\Phi^{\dagger}))$ which, since $\Phi^{C\dagger}$ is unital, by Theorem \ref{phiphidagger} necessarily contains the only other algebra we might naturally associate with $\Phi^{C\dagger}$: $MD(\Phi^C)$. We have the inclusions
\begin{equation}\label{inclusions} MD(\Phi^C)\subseteq \Phi^{C\dagger}(MD(\Phi^{C\dagger}))\subseteq S_{\Phi}\subseteq S_{\Phi}^{\prime\prime}=MD(\Phi)^{\prime}\end{equation}
which suggests the following natural necessary condition for an algebra $\mathcal{A}$ to be privatized by $\Phi^C$: 
\begin{thm} Let $\Phi:M_n(\C)\rightarrow M_m(\C)$ be unital and trace-preserving, so that $\Phi^C$ is trace-preserving and $\Phi^{C\dagger}$ is unital. If $\mathcal{A}$ is a $*$-subalgebra privatized by $\Phi^C$, necessarily $\mathcal{A}$ is quasiorthogonal to $\Phi^{C\dagger}(MD(\Phi^{C\dagger}))$. 
\end{thm}
\begin{proof} This follows from the second inclusion in the chain of inclusions \ref{inclusions}.
\end{proof}

Finally, we seek to make a connection between private algebras and correctable algebras.

\begin{thm} Let $\Phi:M_n(\C)\rightarrow M_m(\C)$ be a unital, trace-preserving completely positive map. Then $MD(\Phi)$ is the subspace of $C^{n\times n}$ on which $\Phi$ acts as a unitary.
\end{thm}
\begin{proof} Since $\Phi$ is both unital and trace-preserving, it is a contraction. Hence, $M_n(\C)\simeq C^{n\times n}$ splits into two subspaces, $M_n(\C) = U_{\Phi}\oplus CNU_{\Phi}$ in such a way that 
$\Phi \bigl|_{U_{\Phi}}$ is unitary and $\Phi \bigl|_{CNU_{\Phi}}$ is completely non-unitary. \\
By definition, $U_{\Phi}$ is the set on which $\Phi^{\dagger}\circ\Phi$ acts as the identity: the fix point set of $\Phi^{\dagger}\circ\Phi$. Hence, by Theorem \ref{MDtpunital}, 
\begin{equation}MD(\Phi) = \mathrm{Fix}(\Phi^{\dagger}\circ\Phi) = U_{\Phi}\end{equation} and we are done. 
\end{proof}
$MD(\Phi)$ is a $*$-algebra, and so is unitarily equivalent to a direct sum 
\begin{equation} MD(\Phi) \simeq \bigoplus_{k=1}^m I_{i_k}\otimes M_{j_k}(\C).\end{equation}
Moreover, recalling Theorem \ref{phiphidagger}, where now $\Phi$ and $\Phi^{\dagger}$ are both trace-preserving, we conclude that $\Phi$, when restricted to $MD(\Phi)$, is a unital $*$-automorphism into $MD(\Phi^{\dagger})$, with inverse $\Phi^{\dagger}$. Hence $MD(\Phi)\simeq MD(\Phi^{\dagger})$. Since both $\Phi$ and $\Phi^{\dagger}$ are trace-preserving, $MD(\Phi)$ and $MD(\Phi^{\dagger})$ are isomorphic, not just automorphic. Hence there exist unitaries $U,V$ such that
\begin{equation} U\bigl(MD(\Phi)\bigr)U^* = V\bigl(MD(\Phi^{\dagger})\bigr)V^* = \bigoplus_{k=1}^m I_{i_k}\otimes M_{j_k}(\C).\end{equation}
Since $\Phi(MD(\Phi))=MD(\Phi^{\dagger})$, $\Phi(MD(\Phi)) = VU^*\bigl(MD(\Phi)\bigr)UV^*$.  \\
Hence, $MD(\Phi)$ is the set of unitarily correctable elements for $\Phi$: those $A \in M_n(\C)$ such that there exists a unitary $W$ such that $A=W^*\Phi(A)W$. If $\mathcal{A}$ is an algebra satisfying $A=W^*\Phi(A)W$ for all $A\in \mathcal{A}$, then necessarily $\mathcal{A}$ is a subalgebra of $MD(\Phi)$. \\

So, for a trace-preserving and unital CP map $\Phi$, the unitarily correctable algebras are the subalgebras of $MD(\Phi)$, while necessary and sufficient conditions for private channels for $\Phi^C$ are that an algebra be quasiorthogonal to $\Phi^{C\dagger}(MD(\Phi^{C\dagger}))$ and quasiorthogonal to $MD(\Phi)^{\prime}$ respectively.\\ 
So then, for $\Phi$ trace-preserving and unital, we have some relations between unitarily correctable algebras for $\Phi$ and private algebras for the complement, $\Phi^C$.
\subsection{Schur Product Maps}
Consider the example of a Schur product channel, $\Phi: M_n(\C)\rightarrow M_n(\C)$, $\Phi(X)=X\circ C$ where $C$ is an $n\times n$ correlation matrix: a positive semidefinite matrix with $1$s down the diagonal. Clearly, $\Phi$ is unital and trace-preserving. If $C=\sum_{i=1}^{p}\lambda_i \hat{v}_i \hat{v}^*_i$ is a spectral decomposition, then, letting $v_i = \sqrt{\lambda_i}\hat{v}_i$, the Kraus operators of $\Phi$ are $\{K_i = \mathrm{diag}(v_i)\}_{i=1}^p$. \\
Form the matrix $V = \sum_{i=1}^p v_ie_i^*$ whose $i^{th}$ column is $v_i$; then $C=VV^*$. Label the columns of $V^*$ by $\{w_i\}_{i=1}^n$, so that $c_{ij} = w_i^*w_j$, and $C$ is the Gram matrix of the vectors $w_i$. Then, by Proposition \ref{opsofcomp} the Kraus operators for $\Phi^{C}$ are $\{L_i = w_ie_i\}_{i=1}^n$. Then the Kraus operators of $\Phi^{C\dagger}$ are $\{L_i^*\}_{i=1}^n$, so 
\begin{equation}\label{schurdagger} \Phi^{C\dagger}(X) = \sum_{i=1}^n e_i w_i^* X w_i e_i = \sum_{i=1}^n (\langle w_i, Xw_i\rangle) E_{ii}.\end{equation}\\

\begin{thm}\label{MDSchur} The multiplicative domain of a Schur product channel $\Phi(X) = X\circ C$ for some correlation matrix $C$ is the algebra $\mathcal{A} = \{X: x_{ij}\neq 0 \textnormal{ iff } |c_{ij}|^2=1\}$.
\end{thm}
\begin{proof}
By Theorem \ref{MDfix}, since $\Phi$ is unital and trace-preserving, its multiplicative domain is equal to the fixed point set $\mathrm{Fix}(\Phi^{\dagger}\circ \Phi).$ 
\begin{align} \Phi^{\dagger}(X) & = \sum_{i=1}^p \mathrm{diag}(\overline{v}_i) X \mathrm{diag}(v_i)\\
& = \sum_{i,j=1}^n x_{ij}E_{ij}(\sum_{k=1}^p \overline{v}_{ki}v_{kj}) \\
& = X\circ (\overline{C}).\end{align}
Therefore, \begin{equation} (\Phi^{\dagger}\circ\Phi)(X) = X\circ(C\circ\overline{C})\end{equation}
and the fixed points of this are matrices $X$ with non-zero entries only at indices $(i,j)$ where $\overline{C}\circ C$ has the entry $1$--when $|c_{ij}|^2 =1$. \\
\end{proof}

\begin{prop}\label{onezerosofC} Let $G$ be the graph on $n$ vertices where $(i,j)\in E(G)$ if and only if $|c_{ij}|^2=1$. 
Then $G$ is the union of $m$ complete graphs on $k_i$ vertices, where $\sum_{i=1}^m k_i = n$. 
\end{prop}
\begin{proof} We will prove this by showing that the graph on any connected component of $G$ must be complete. Let $G_i$ be a connected component. If $|G_i|=1,2$, it is trivially a complete graph. So, consider the case that $|G_i|=3$. Then, there exist vertices $i,j,k\in G_i$ such that $(i,j), (j,k)\in G_i$. Let $C[i,j,k]$ be the principle submatrix of $C$ on the indices $i,j,k$.\\ 

Since $C$ is positive semidefinite, every principle submatrix must be positive semidefinite as well. Then 
\begin{equation} C[i,j,k] = \begin{pmatrix} 1 & c_{ij} & c_{ik} \\ \overline{c}_{ij} & 1 & c_{jk} \\ \overline{c}_{ik} & \overline{c}_{jk} & 1 \end{pmatrix} \geq 0.\end{equation}

Taking the Schur complement, this is equivalent to 
\begin{equation} \begin{pmatrix} 1 & c_{jk} \\ \overline{c}_{jk} & 1 \end{pmatrix} - \begin{pmatrix} |c_{ij}|^2 & c_{ik}\overline{c}_{ij} \\ c_{ij}\overline{c}_{ik} & |c_{ik}|^2\end{pmatrix} = \begin{pmatrix} 0 & c_{jk}-\overline{c}_{ij}c_{ik} \\ \overline{c}_{jk}-c_{ij}\overline{c}_{ik} & 1-|c_{ik}|^2 \end{pmatrix} \geq 0\end{equation} since $|c_{ij}|^2 = |c_{jk}|^2 = 1$. But this necessitates $c_{ik} = \frac{c_{jk}}{\overline{c}_{ij}}$ and so $|c_{ik}|^2 =1$ as well. \\
Hence, if $|G_i|=3$, $G_i$ is complete.\\
If $|G_i|=n$, and all subgraphs on $n-1$ vertices are complete, $G_i = K_n$. So, by induction, we are done. 
\end{proof}
Hence, up to a permutation that relabels the vertices of $G$ in a way consistent with the decomposition into connected components, $MD(\Phi) = \bigoplus_{i=1}^m M_{k_i}(\C)$.\\
Hence, the algebra generated by $S_{\Phi}$ is $MD(\Phi)^{\prime} = \bigoplus_{i=1}^m \C I_{k_i}$.\\
\begin{prop} Let $\Phi(X)=X\circ C$ be as above, and let $\{w_i\}_{i=1}^n$, $w_i\in \C^p$ be defined as above, so that $c_{ij} = \langle w_i,w_j\rangle$. \\

The multiplicative domain of $\Phi^{C\dagger}$ is the set of matrices $A$ such that all $w_i$ are eigenvectors of $A^*$.
\end{prop}
\begin{proof} Recall Equation \ref{schurdagger}, 
$$\Phi^{C\dagger}(X) = \sum_{i=1}^n \langle w_i,Xw_i\rangle E_{ii}.$$
Since the range of this map is diagonal, we can consider each diagonal entry at a time. For $A\in MD(\Phi^{C\dagger})$, we require

$$\langle w_i,Aw_i\rangle \langle w_i,Xw_i\rangle = \langle w_i,AXw_i\rangle$$ for all $X$. 
Hence \begin{align} \langle w_i, \langle w_i,Aw_i\rangle Xw_i-AXw_i\rangle & = \langle w_i, \bigl(\langle w_i,Aw_i\rangle I - A\bigr)Xw_i\rangle \\
& = \langle w_i, \bigl(\langle w_i,Aw_i\rangle I - A\bigr)v\rangle \\
& = 0.\end{align}
Hence, expressing $A$ and $v$ an any orthonormal basis $\{x_i\}$ where $x_1=w_i$, 
$$\sum_{j=2}^p a_{1j}v_j = 0$$ for all $v$, and so in such a basis
$$A = \begin{pmatrix} a_{11} & 0 & \ldots & 0 \\ a_{21} & a_{22} & \ldots & a_{2p} \\ \vdots & \vdots & \ddots & \vdots \\ a_{p1} & a_{p2} & \ldots & a_{pp}\end{pmatrix},$$ proving the claim. 
\end{proof}
Let $A \in MD(\Phi^{C\dagger})$ and $A^*w_i = \lambda_i$. Then 
\begin{equation} \Phi^{C\dagger}(A) = \sum_{i=1}^n \langle w_i,Aw_i\rangle E_{ii} = \sum_{i=1}^n \lambda_i E_{ii}.\end{equation}
Recall that $MD(\Phi)$ is, up to a permutation of $C$, $\bigoplus_{i=1}^m M_{k_i}(\C)$, where the $k^{th}$ diagonal $k_i\times k_i$ block of $C$ is a principal submatrix consisting only of entries with modulus $1$. \\
For $(i,j)$ in this block, we have that 
\begin{align} \langle w_i,w_i\rangle & = c_{ii} =1\\
|\langle w_i,w_j\rangle| & = |c_{ij}| =1\end{align}
and so, by Cauchy-Schwarz, $w_i = z_{ij}w_j$ where $z_{ij}$ is a complex number of modulus one.\\
Hence, for $A\in MD(\Phi^{C\dagger})$, if $A^*w_i = \lambda_i w_i$ then necessarily $A^*w_j = \lambda_i w_j$. Hence, $\Phi^{C\dagger}(A) = \bigoplus_{k=1}^m \lambda_k I$ where the direct sum decomposition is the same decomposition as for $MD(\Phi)$. Hence, $$\Phi^{C\dagger}(MD(\Phi^{C\dagger}))\subseteq MD(\Phi)^{\prime} = S_{\Phi}^{\prime\prime}.$$ \\
If the repetitions of $w_i$ enforced by the modulus-$1$ pattern of $C$ are the only instance of linear dependence in the set $\{w_i\}_{i=1}^n$, then the two sets are equal, and we obtain a necessary and sufficient condition for $\Phi^C$ to privatize an algebra $\mathcal{A}$: $\mathcal{A}$ is quasiorthogonal to $MD(\Phi)^{\prime} = \bigoplus_{i=1}^m \C I_{k_i}$.\\
That $C$ decomposes into $m$ principal diagonal blocks, each of which is rank $1$, means that $\mathrm{rank}(C)\leq m$. In the case that $\mathrm{rank}(C)=m$, we have no non-trivial linear dependencies among the $\{w_i\}$. If $\mathrm{rank}(C)< m$, then we have more linear dependencies, and the containment of the two algebras is strict. \\
We can be more precise about the form of $\Phi^{C\dagger}(MD(\Phi^{C\dagger}))$. Label the $m$ $k_i\times k_i$ blocks of all modulus-$1$ entries $C_k$. Then there exists some $v_k\in \C^{k_i}$ such that $C_k = v_kv_k^*$ and $v_k$ has all entries of modulus $1$. Then $C$ has the form
$$C = \begin{bmatrix} v_1v_1^* & A_{12} & \ldots & A_{1m} \\ A_{12}^* & v_2v_2^* & \ldots & A_{2m} \\
\vdots & \vdots & \ddots & \vdots \\ A_{1m}^*& A_{2m}^* & \ldots & v_mv_m^*\end{bmatrix}.$$
\begin{prop} For $C$ as above, each $A_{ij}$ must have the form 
$$A_{ij} = a_{ij} v_iv_j^*.$$
\end{prop}
\begin{proof} Start with the two upper left blocks. $C\geq 0$ so 
$$\begin{bmatrix} v_1v_1^* & A_{12} \\ A_{12}^* & v_2v_2^* \end{bmatrix} \geq 0.$$
Introduce the orthonormal basis $\{w_i\}_{i=1}^{k_1}$ with $w_1=\frac{1}{\sqrt{k_1}}v_1$, and the basis $\{e_i\}_{i=1}^{k_2}$ the standard basis on $\C^{k_2}$. Then $I= \sum_{i=1}^{k_1} w_iw_i^*$, and so taking the Schur complement, and using $\lim_{x\rightarrow 0}(v_1v_1^* + xI)^{-1}$ (since $v_1v_1^*$ is not invertible), we have
\begin{equation} v_2v_2^* - A_{12}^*(v_1v_1^*+x\sum_{i=1}^{k_1} w_iw_i^*)^{-1}A_{12}\geq 0.\end{equation}
Then \begin{align} (v_1v_1^* + x\sum_{i=1}^{k_1}w_iw_i^*)^{-1} & = \frac{1}{k_1+x}w_1w_1^*+\frac{1}{x}\sum_{i=2}^{k_1}w_iw_i^*.\end{align}

Express $A_{12}, A_{12}^*$ in terms of the bases $\{w_i\}, \{e_i\}$ to obtain
\begin{align*} A_{12}^*(v_1v_1^*+x\sum_{i=1}^{k_1}w_iw_i^*)^{-1}A_{12} & = \sum_{i,j=1}^{k_p} \overline{a}_{ji} e_iw_j^*\bigl( \frac{1}{k_1+x}w_1w_1^* + \frac{1}{x}\sum_{k=2}^{k_1}w_kw_k^*\bigr)\sum_{r,s} a_{rs} w_re_s^* \\
& = \sum_{i,j\neq 1,k\neq 1,s} \frac{\overline{a}_{ji}a_{js}}{x} e_ie_s^* + \sum_{i,s} \frac{1}{k_1+x}\overline{a}_{1i}a_{1s}e_ie_s^*.\end{align*}
Then, the $i^{th}$ diagonal entry of $v_2v_2^*$ is $1$, so the Schur complement being positive semidefinite requires 
\begin{equation} 1 \geq \frac{\sum_{j\neq 1}|a_{ji}|^2}{x}+\frac{|a_{1i}|^2}{k_1+x}, \ \ x\rightarrow 0.\end{equation}  
Clearly, the first term on the right hand side blows up, and so $|a_{ji}|^2 = 0$ for all $j\neq 1$. Hence $A_{12}$ when expressed in the basis $e_iw_j^*$ has only its first row non-zero, $a_{1i} \leq i_1$. Hence $A_{12}=\sum_i a_{1i} w_1 e_j^* = w_1\sum_j a_{1j}e_j^* = \frac{1}{\sqrt{k_1}}v_1\sum_j a_{1j}e_j^*=a_{12}v_1 v^*$ for some vector $v$.  \\
A similar analysis taking the Schur complement the other way shows that $A_{12} = uv_2^*$ for some $u$, and s $A_{12} = a_{12}v_1v_2^*$.\\   
Now, assume that $A_{ij} = a_{ij}v_iv_j^*$ for an $(n-1)\times (n-1)$ block matrix matrix. Then, for a matrix of block-size $n\times n$ to have the prescribed form and be positive semdifinite requires

$$C = \begin{bmatrix} v_1v_1^* & a_{12}v_1v_2^* & \ldots & A_{1n}^* \\
a_{12}^* & v_2v_2^* & \ldots & a_{2n}v_2v_n^* \\
\vdots  & \vdots & \ddots & \vdots \\
A_{1n}^* & \overline{a}_{2n} v_2v_n^* & \ldots & v_nv_n^*\end{bmatrix}.$$
Taking the Schur complement 
$$\begin{bmatrix} v_2v_2^* & a_{23}v_2v_3^* & \ldots & a_{2n}v_2v_n^* \\ 
\overline{a}_{23}v_3v_2^* & v_3 v_3^* & \ldots & a_{3n}v_3v_n^* \\
\vdots & \vdots & \ddots & \vdots \\
\overline{a}_{2n} v_nv_2^* & \overline{a}_{3n} v_nv_3^* & \ldots & v_nv_n^*\end{bmatrix} - \begin{bmatrix} \overline{a}_{12}v_2v_1^* \\ \overline{a}_{13}v_3v_1^* \\ \vdots \\ A_{1n}^* \end{bmatrix} \biggl(v_1v_1^* + xI\biggr)^{-1} \begin{bmatrix} a_{12}v_1v_2^* & a_{12}v_1v_3^* & \ldots & A_{1n}\end{bmatrix} \geq 0$$ for all $x\rightarrow 0$.\\
This positivity requires the bottom-right block to be positive, which in turn requires $$v_nv_n^* - A_{1n}^* \biggl(v_1v_1^* + xI\biggr)^{-1}A_{1n} \geq 0$$ as $x\rightarrow 0$.
So, by essentially the same analysis as above, $A_{1n} = a_{1n}v_1v_n^*$.
\end{proof}

Let $$C=\begin{bmatrix} v_1v_1^* & a_{12}v_1v_2^* & \ldots & a_{1m}v_1v_m^* \\ \overline{a}_{12} v_2v_1^* & v_2v_2^* & \ldots & a_{2m}v_2v_m^* \\
\vdots & \vdots & \ddots & \vdots \\
\overline{a}_{1m}v_mv_1^* & \overline{a}_{2m}v_mv_2^* & \ldots & v_mv_m^*\end{bmatrix}$$ and define the matrix $A$ by $A_{ii}=1$, $A_{ij} = a_{ij}$ for $1\leq i,j\leq m$. This is a correlation matrix, and $\mathrm{rank}(C)=\mathrm{rank}(A)$. Necessarily, $a_{ij}<1$ for $i\neq j$, and hence $\mathrm{rank}(A)>1$ so long as $m>1$. Also, clearly $\mathrm{rank}(A)\leq m$. \\
Then, linear dependencies imposed on the Gram vectors of $C$ beyond the ones required to give $C$ its modulus-$1$ pattern are linear dependencies in the columns of $A$. This information is contained in the independence matroid of the columns of $A$: if $A$ has columns $A_1,\ldots,A_m\in \C^m$ then $\mathrm{Mat}(A)$ is the subsets of $\{A_1,\ldots,A_m\}$ where a subset $S=\{A_{i_1},A_{i_2},\ldots,A_{i_l}\}$ is independent if and only if $A_{i_1},A_{i_2},\ldots,A_{i_l}$ are linearly independent. The cycles of $\mathrm{Mat}(A)$ are the minimal dependent sets. If $C = \{A_{i_1},\ldots,A_{i_l}\}$ is a cycle of $\mathrm{Mat}(A)$ then the image of any member of $MD(\Phi^{C\dagger})$ under $\Phi^{C\dagger}$ is $\bigoplus_{i=1}^m \lambda_i I_{k_i}$ with $\lambda_{i_1}=\lambda_{i_2}=\ldots = \lambda_{i_l}$. \\
Hence, the cycles of this matroid control the subalgebra $\Phi^{C\dagger}(MD(\Phi^{C\dagger}))$ for such a channel, and determine when it will equal to $S_{\Phi}^{\prime\prime}$.

\bibliographystyle{plain}
\bibliography{uncertaintybib}

\begin{thebibliography}{10}

\bibitem{ambainis}
Andris Ambainis, Michele Mosca, Alain Tapp, and Ronald De~Wolf.
\newblock Private quantum channels.
\newblock In {\em focs}, pages 547--553, 2000.

\bibitem{boykin}
P~Oscar Boykin and Vwani Roychowdhury.
\newblock Optimal encryption of quantum bits.
\newblock {\em Physical review A}, 67(4):042317, 2003.

\bibitem{choilinear}
Man-Duen Choi.
\newblock Completely positive linear maps on complex matrices.
\newblock {\em Linear algebra and its applications}, 10(3):285--290, 1975.

\bibitem{choischwarz}
MD~Choi.
\newblock A schwarz inequality for positive linear maps on c*-algebras.
\newblock 1974.

\bibitem{chruscinski}
Dariusz Chru{\'s}ci{\'n}ski.
\newblock Positive maps, doubly stochastic matrices and new family of spectral
  conditions.
\newblock In {\em Journal of Physics: Conference Series}, volume 213, page
  012003. IOP Publishing, 2010.

\bibitem{deaett}
Louis Deaett.
\newblock The minimum semidefinite rank of a triangle-free graph.
\newblock {\em Linear Algebra and its Applications}, 434(8):1945--1955, 2011.

\bibitem{devetak}
Igor Devetak.
\newblock The private classical capacity and quantum capacity of a quantum
  channel.
\newblock {\em IEEE Transactions on Information Theory}, 51(1):44--55, 2005.

\bibitem{devshor}
Igor Devetak and Peter~W Shor.
\newblock The capacity of a quantum channel for simultaneous transmission of
  classical and quantum information.
\newblock {\em Communications in Mathematical Physics}, 256(2):287--303, 2005.

\bibitem{jochym}
Tomas Jochym-O’Connor, David~W Kribs, Raymond Laflamme, and Sarah Plosker.
\newblock Private quantum subsystems.
\newblock {\em Physical review letters}, 111(3):030502, 2013.

\bibitem{plosker}
Tomas Jochym-O'Connor, David~W Kribs, Raymond Laflamme, and Sarah Plosker.
\newblock Quantum subsystems: Exploring the complementarity of quantum privacy
  and error correction.
\newblock {\em Physical Review A}, 90(3):032305, 2014.

\bibitem{kraus}
K~Kraus.
\newblock {\em States, Effects and Operations: Fundamental Notions of Quantum
  Theory}.
\newblock Springer-Verlag, 1983.

\bibitem{kks}
Dennis Kretschmann, David~W Kribs, and Robert~W Spekkens.
\newblock Complementarity of private and correctable subsystems in quantum
  cryptography and error correction.
\newblock {\em Physical Review A}, 78(3):032330, 2008.

\bibitem{kribslaflamme}
David Kribs, Raymond Laflamme, and David Poulin.
\newblock Unified and generalized approach to quantum error correction.
\newblock {\em Physical review letters}, 94(18):180501, 2005.

\bibitem{kribswavelets}
David~W Kribs.
\newblock Quantum channels, wavelets, dilations and representations of
  $\mathcal{O}_{n}$.
\newblock {\em Proceedings of the Edinburgh Mathematical Society (Series 2)},
  46(02):421--433, 2003.

\bibitem{kribsspekkens}
David~W Kribs and Robert~W Spekkens.
\newblock Quantum error-correcting subsystems are unitarily recoverable
  subsystems.
\newblock {\em Physical Review A}, 74(4):042329, 2006.

\bibitem{levick}
Jeremy Levick, Tomas Jochym-O’Connor, David~W Kribs, Raymond Laflamme, and
  Rajesh Pereira.
\newblock Private quantum subsystems and quasiorthogonal operator algebras.
\newblock {\em Journal of Physics A: Mathematical and Theoretical},
  49(12):125302, 2016.

\bibitem{minc}
Henryk Minc.
\newblock {\em Nonnegative matrices}.
\newblock Technion-Israel Institute of Technology, Dept. of Mathematics, 1974.

\bibitem{paulsen}
Vern Paulsen.
\newblock {\em Completely bounded maps and operator algebras}, volume~78.
\newblock Cambridge University Press, 2002.

\bibitem{petz}
D{\'e}nes Petz.
\newblock Algebraic complementarity in quantum theory.
\newblock {\em Journal of Mathematical Physics}, 51(1):015215, 2010.

\bibitem{pudlak}
Pavel Pudl{\'a}k.
\newblock Cycles of nonzero elements in low rank matrices.
\newblock {\em Combinatorica}, 22(2):321--334, 2002.

\bibitem{stinespring}
W~Forrest Stinespring.
\newblock Positive functions on c*-algebras.
\newblock {\em Proceedings of the American Mathematical Society},
  6(2):211--216, 1955.

\end{thebibliography}
\end{document}